%% file: main.tex
\documentclass[letterpaper, 10 pt, conference]{ieeeconf}  

\IEEEoverridecommandlockouts                              

\usepackage{graphicx}
\usepackage{amsmath}
\usepackage{amssymb}
\usepackage{amsthm}
\usepackage[english]{babel}
\usepackage[utf8]{inputenc}
\usepackage{algorithm}
\usepackage[noend]{algpseudocode}
\usepackage{cleveref}
\usepackage{times} 
\usepackage{mathtools}
\usepackage{booktabs}
\usepackage{bm}
\usepackage{wrapfig}
\usepackage{xcolor}

\input{math_macros}

\title{\LARGE \bf
Safety Index Synthesis with State-dependent Control Space
}

\author{Rui Chen$^{1}$, Weiye Zhao$^{1}$, and Changliu Liu$^{1}$
\thanks{$^{1}$Carnegie Mellon University, Pittsburgh, PA. Contact: {\tt\small \{ruic3, weiyezha, cliu6\}@andrew.cmu.edu}}%
\thanks{This work is partially supported by the National Science Foundation, Grant No. 2144489.}%
}

\begin{document}

\maketitle
\thispagestyle{empty}
\pagestyle{empty}

\begin{abstract}

This paper introduces an approach for synthesizing feasible safety indices to derive safe control laws under state-dependent control spaces.
The problem, referred to as Safety Index Synthesis (SIS), is challenging because it requires the existence of feasible control input in all states and leads to an infinite number of constraints.
The proposed method leverages Positivstellensatz to formulate SIS as a nonlinear programming (NP) problem. 
We formally prove that the NP solutions yield safe control laws with two imperative guarantees: forward invariance within user-defined safe regions and finite-time convergence to those regions.
A numerical study validates the effectiveness of our approach.

\end{abstract}

\section{Introduction}




Energy functions \cite{wei2019unified} are widely used in literature to quantify safety and derive safe control strategies.
Representative functions mainly include safety indices (SI) for the safe set algorithm (SSA) \cite{liu2014control} and control barrier functions (CBF) \cite{ames2014control}.
Safety indices can often derive constraints affine on control which yields safe control laws as the solution to quadratic programming (QP) problems.
To achieve provable safety, the system guarded by the safety index should secure two critical properties: (a) \textbf{\textit{forward-invariance}}, meaning that the system should stay in a safe region once entering it, and (b) \textbf{\textit{finite-time convergence}}, meaning that the system should land in the safe region in finite time, even when starting in an unsafe state.
Moreover, those guarantees can hold only when the control constraints derived from the SI are always feasible in the first place.
Namely, in every state of interest, there must exist a control in the control space (either bounded or unbounded), that satisfies the safety constraints.
The feasibility defined above is often realized by properly designing the safety index, which we refer to as the problem of \textit{Safety Index Synthesis} (SIS) \cite{zhao2021zeroviolation, zhao2023sos}.
In that regard, SIS is similar to synthesizing CBFs for enforcing constraints \cite{ames2014control}, designing value functions in Hamilton-Jacobi (HJ) reachability analysis \cite{choi2021robust}, and constructing control Lyapunov functions for system stabilization \cite{freeman2008robust}.
However, SIS is different from those approaches in that the desired safety index refers to a specific class of energy functions normally used for collision avoidance with SSA \cite{liu2014control, lin2017real}.

The problem of SIS is non-trivial since it requires feasibility over a continuous state space, yielding an infinite number of constraints.
Recent approaches tackle this challenge either approximately on sampled states \cite{ma2022joint, dawson2022safe, wei2022nndm, liu2022inputsat} or exactly on all state via manual synthesis \cite{zhao2021zeroviolation}, Hamilton-Jacobi (HJ) reachability \cite{choi2021robust}, or sum-of-square programming (SOSP) \cite{zhao2023sos, kang2023verification}.
Although the above approaches are promising, they consider unbounded control or known constant control limits.
Regarding unbounded control, \cite{dawson2022safe} learns an approximate CBF using neural networks but without theoretical safety guarantees.
Regarding bounded control, \cite{ma2022joint} optimizes the safety index using gradient-based methods without safety guarantees.
\cite{zhao2021zeroviolation} manually derives a synthesis rule for simplified unicycle kinematics, which is hard to generalize.
\cite{wei2022nndm} synthesizes safety indices for neural network dynamics via evolutionary search and is hard to scale.
\cite{zhao2023sos} transforms SIS into nonlinear programming via SOSP.
However, \cite{zhao2023sos} only guarantees forward-invariance, missing the property of finite-time convergence.
Regarding CBF synthesis, \cite{liu2022inputsat} learns to minimize the risk of control saturation on counterexamples, \cite{choi2021robust} synthesizes CBF utilizing Hamilton-Jacobi reachability, and \cite{kang2023verification} focuses on bounded additive uncertainty and convex polynomial constraints on control.

In practice, the control of a system is not only always bounded, but sometimes \textit{state-dependent} instead of constant.
While constant control space applies to many systems such as articulated robots and aerial robots with bounded force and torque inputs, state-dependent control space is more common when we model and control the system on a kinematic level.
For example, when modeling a quadruped robot as a second-order unicycle, the maximal acceleration at different velocities varies due to hardware limits (e.g., zero forward acceleration at top speed but non-zero acceleration when still).
On the other hand, kinematic modeling often comes in simple forms (e.g., integrators, unicycles) which facilitates control law derivations.
Hence, it is valuable to derive safe control for systems with state-dependent control spaces.

In this paper, our contributions are: \textbf{(a) a novel nonlinear programming approach to exact SIS under state-dependent control space with (b) formal proof of both forward invariance with user-defined safety regions and finite-time convergence to those regions}.
In the rest of this paper, we first formulate the problem of SIS in \Cref{sec:problem}.
Then, we derive the NP problem in \Cref{sec:SIS} with formal proofs in \Cref{sec:theory}.
Finally, we present a numerical study in \Cref{sec:exp} and conclude with future work in \Cref{sec:conclusion}.





\section{Problem Formulation}\label{sec:problem}

\subsection{Dynamic System}\label{sec:dynamic_system}

We consider a dynamic system with state-dependent control limits.
Let $x\in\cX\subset \RR^{N_x}$ be the system state where the state space $\cX$ is characterized by a set of constraints $\cX\defeq\{x\mid h_i(x) \geq 0, \forall i=1,\dots,N_h\}$.
Let $u\in\cU(x)$ be the control input where $\cU(x)\defeq\{u\in\RR^{N_u} \mid \underline{u}(x) \leq u \leq \bar{u}(x) \}$.
$\underline{u}:\RR^{N_x}\mapsto\RR^{N_u}$ and $\bar{u}:\RR^{N_x}\mapsto\RR^{N_u}$ define element-wise limits on $u$.
Notably, $\cU(x)$ is \textbf{state-dependent}.
As explained in the previous section, such an assumption is useful when we need to model a dynamic system using a surrogate model, for instance modeling a quadruped robot as a unicycle (see \Cref{sec:exp} for an example).
The dynamics is then given by
\begin{equation}
    \dot{x} = f(x) + g(x)u, ~ u \in \cU(x),
\end{equation}
where $f: \RR^{N_x} \mapsto \RR^{N_x}$ and $g: \RR^{N_x} \mapsto \RR^{N_x\times N_u}$ are both locally Lipschitz continuous.



\subsection{Safe Control Preliminaries}\label{sec:formulation_safe_control}

\paragraph{Safety Specification}
The safety specification requires the system state to be constrained within a closed subset $\cX_S$ (i.e., \textit{spec set}) of the state space $\cX$.
We assume that $\cX_S$ is equivalent to the zero sublevel set of some piecewise smooth function $\phi_0\defeq \cX \mapsto \RR$, i.e., $\cX_S \defeq \{x\in \cX \mid \phi_0(x) \leq 0\}$.
Both $\cX_S$ and $\phi_0$ should be designed by users.
For example, to keep the distance $d$ to some obstacle above a threshold $d_\textrm{min}$, $\phi_0$ can be given as $\phi_0 = d_\mathrm{min} - d$.

\paragraph{Safe Control Objectives}
\revcolor{
First, we need to find a subset $\cX_\mathrm{safe}$ (i.e., \textit{safe set}) of $\cX_S$ such that if the state $x$ is already within $\cX_\mathrm{safe}$, it should never leave that set.
}
Formally, $x(t_0) \in \cX_\mathrm{safe} \rightarrow x(t) \in \cX_\mathrm{safe}, ~ \forall t \geq t_0$.
We refer to such property as \textbf{\textit{forward invariance (FI)}}.
Secondly, if the state $x$ is outside the safe set, it should land in the safe set in finite time.
Formally, $\forall t_0 ~\st~ x(t_0)\in \cX\setminus\cX_\mathrm{safe},~\exists t\in[t_0,\infty) \rightarrow x(t) \in \cX_\mathrm{safe}$.
We refer to such property as \textbf{\textit{finite-time convergence (FTC)}}.

\paragraph{Safe Control Backbone}
\revcolor{
There are states in $\mathcal{X}_S$ that will inevitably go to unsafe regions due to either actuation limits or dynamic limits (i.e., high relative degree).
This section handles dynamic limits, while actuation limits will be addressed in \Cref{sec:SIS_formulation}.
With high relative degree, $u$ does not appear in $\dot\phi_0$.
For instance, for a second-order system, $\dot\phi_0 =  - \dot d$ does not directly depend on the acceleration input.
Hence,}
we need an alternative safety quantification $\phi$ that considers the system dynamics.
For instance, for a second-order system (e.g., acceleration-driven), we can employ a first-order safety index $\phi$ (e.g., velocity-based), such that the derivative $\dot{\phi}$ depends on the input and can be used to bound $u$ to prevent unsafe states.
The safe set algorithm (SSA) \cite{liu2014control} provides a systematic approach to the design of such $\phi$, and we adopt SSA as our backbone.
SSA introduces a continuous, piece-wise smooth energy function $\phi \defeq \cX \mapsto \RR$ (i.e., \textbf{\textit{safety index}}) to quantify the safety level of states considering the system dynamics.
The general form of an $n^\mathrm{th}$ ($n\geq 0$) order safety index $\phi_n$ is given as\footnote{The general form can also include nonlinear shaping of the safety index \cite{zhao2021zeroviolation,wei2022nndm}. This paper focuses on the linear form, while the nonlinear shaping under state-dependent control space will be left for future work.}
\begin{equation}\label{def:phi_recursive}
    \begin{aligned}
        \phi_n &= (1+a_1 s)(1+a_2 s)\dots(1+a_n s)\phi_0,
    \end{aligned}
\end{equation}
where $s$ is the differentiation operator.
We expand \eqref{def:phi_recursive} as
\begin{equation}\label{def:phi_root}
    \phi_n \defeq \phi_0 + \textstyle\sum_{i=1}^{n}k_i \phi^{(i)}_0.
\end{equation}
where $\phi_0^{(i)}$ is the $i^\mathrm{th}$ time derivative of $\phi_0$.
$\phi_n$ should satisfy that (a) the roots of the characteristic equation $\prod_{i=1}^n(1+a_i s) = 0$ are all negative real (to avoid overshooting of $\phi_0$), (b) $\phi_0^{(n)}$ has relative degree one to the control input.
For instance, for $\phi_0=d_\mathrm{min} - d$ we can use $\phi_1=d_\mathrm{min} - d - k\dot{d}$ for collision-avoidance for a second-order system.
We will prove in this paper that if we enforce $\dot{\phi}_n(x, u) \leq -\eta$ when $\phi_n(x) \geq 0$ for some constant $\eta > 0$, both forward invariance within a safe set $\cX_\mathrm{safe}$ and finite-time convergence to that set are guaranteed (see \Cref{fig:SIS}).
Hence, the safe control law $c_{\phi_n}$ of SSA can be written as the following optimization:
\begin{equation}\label{eq:safe_control_law}
    \minimizewrt{u\in\cU(x)} \cJ(u) ~ \st ~ \dot{\phi}_n(x,u) \leq -\eta~\mathrm{if}~\phi_n(x) \geq 0 
\end{equation}
where the objective $\cJ$ is arbitrary.
If minimal deviation from some nominal control $u^r$ is desired, one can apply $\cJ(u) = \|u-u^r\|$.
In that case, \eqref{eq:safe_control_law} forms a QP similar to that in \cite{xiao2019control} using high order control barrier functions (HOCBF), with a different constraint in the QP.
Our approach guarantees both FI (as HOCBF does) and FTC (missing from HOCBF).

\begin{figure}
    \centering
    \includegraphics[width=\linewidth]{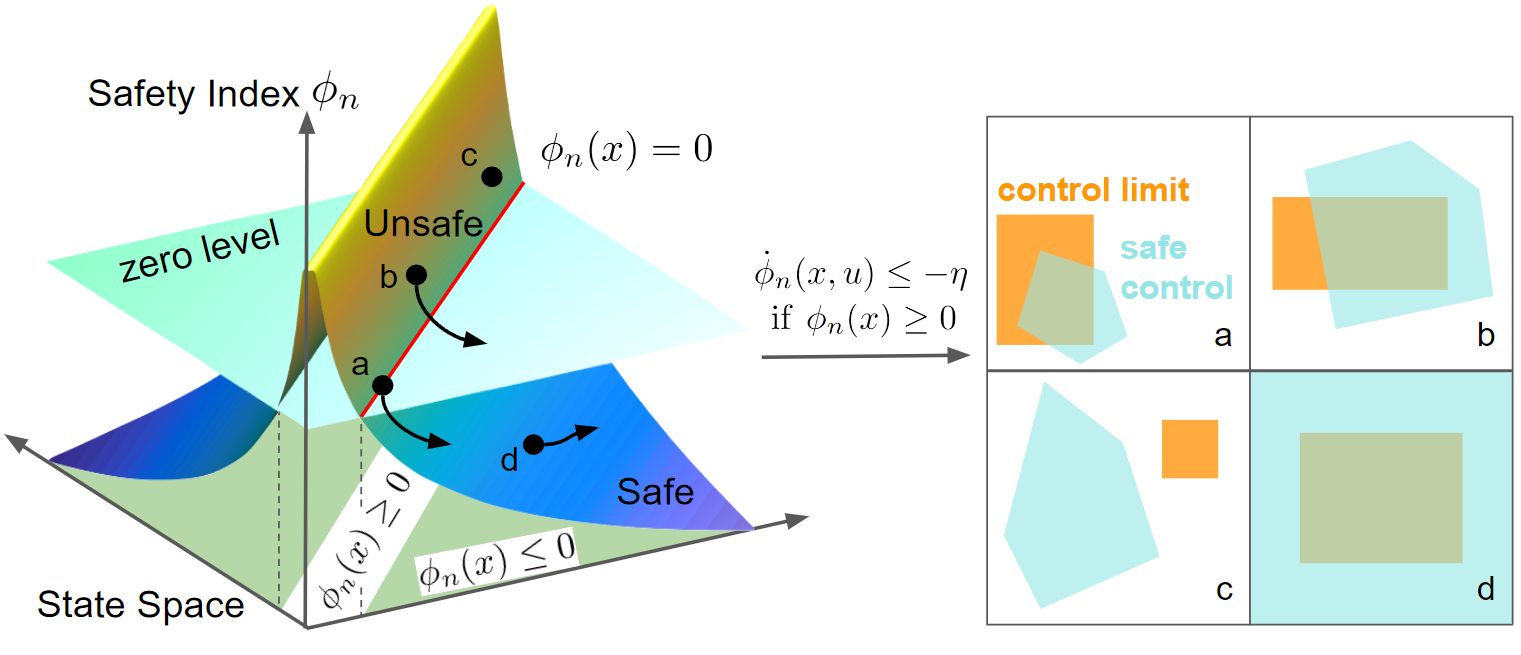}
    \caption{The safety index manifold, control spaces, and safety constraints. Feasible safe control is found in unsafe states a and b but not in c. State d is within the safe set where the safety constraint is inactive. Due to the infeasibility in state c, the SI cannot guarantee safety and is invalid.}
    \label{fig:SIS}
    \vspace{-15pt}
\end{figure}

\subsection{Safe Index Synthesis}\label{sec:SIS_formulation}

To achieve safe control objectives by implementing \eqref{eq:safe_control_law}, the optimization itself must be feasible in the first place.
Namely, we have to ensure that there is always a control within \revcolor{the actuation limit} $\cU(x)$ to reduce positive $\phi_n$ values.
We refer to the above objective as \textit{Safety Index Synthesis} (SIS), formally given by \Cref{problem:synthesis}.
\begin{problem} [Safety Index Synthesis] \label{problem:synthesis}
    Find safety index as $\phi_\theta \defeq \phi_0 + \sum_{i=1}^{n}k_i \phi^{(i)}_0$ with parameter $\theta\in\Theta\defeq\{[k_1,k_2,\dots,k_{n}] \mid k_i\in\RR,k_i\geq 0,\forall i \}$, such that
    \begin{equation}\label{eq:synthesis}
        \forall x\in\cX ~ \st ~ \phi_\theta(x)\geq 0, \minimizewrt{u\in\cU(x)} \dot{\phi}_\theta(x,u) < -\eta
    \end{equation}
\end{problem}
$\phi_\theta$ is the $n^\mathrm{th}$ order safety index parameterized by $\theta$ and is used interchangeably with $\phi_n$ in this paper for clarity.
Note that \Cref{problem:synthesis} is difficult for having infinitely many constraints, since \eqref{eq:synthesis} needs to hold for \textbf{any} state $x~\st \phi_\theta(x) \geq 0$.
Solution to a similar SIS problem can be found in \cite{zhao2023sos}.
\cite{zhao2023sos} only considers a constant $\cU$, while this work handles state-dependent $\cU(x)$.
\cite{zhao2023sos} only considers the existence of safe control on the boundary $\phi_\theta(x)=0$, while we consider all unsafe states $\phi_\theta(x)\geq 0$.
Moreover, \cite{zhao2023sos} only makes $\phi_\theta$ non-increasing, i.e., $\dot{\phi}_\theta \leq 0$ and does not recover from unsafe states.
Our work enforces a strict decrease of $\phi_\theta$ outside safe regions which we will prove to enable FTC.

Overall, compared to previous work, \Cref{problem:synthesis} poses a more general class of SIS problems with stronger safety guarantees.
Our proposed synthesis problem can easily be extended to other safe control backbones such as CBFs \cite{ames2014control} and HOCBFs \cite{xiao2019control}.
Finally, we define a safe control law to be feasible if it is derived from the result of SIS:
\begin{definition}[Feasible Safe Control Law]\label{def:feasible_control_law}
    Given a safety index $\phi_{n}$, we define the control law $c_{\phi_n}$ defined by \eqref{eq:safe_control_law} to be feasible if (i) $a_i>0,\forall i\in[n]$; (ii) $\phi_0$ is $n+1^\mathrm{th}$ order differentiable everywhere\footnote{A more general definition of feasible safe control laws requires $\phi_0$ to be $n+1^\mathrm{th}$ order \textit{piece-wise} differentiable, which is left for future work.}; and (iii) $\phi_n$ solves \Cref{problem:synthesis}.
\end{definition}

\section{Safety Index Synthesis via SOSP}\label{sec:SIS}

We first overview \textit{Positivstellensatz}, a key theorem backing our derivation, and then present the conversion of \Cref{problem:synthesis} into (a) \textit{a local manifold positiveness problem}, (b) \textit{a refute problem} and finally (c) \textit{a nonlinear programming problem}.


\subsection{Preliminary: Positivstellensatz}

Problems of finding a function $f$ such that $\forall x\in\RR^{N_x},~g\circ f(x) > 0$ given some known function $g$ are called \textit{\textbf{global} positiveness problems} \cite{zhao2022provably}.
They can be converted to sum of squares programming (SOSP) problems and solved via SOSTOOLS \cite{papachristodoulou2013sostools}.
However, \Cref{problem:synthesis} limits the constraint satisfaction to a local manifold $\{x\in\cX \mid \phi_\theta(x) \geq 0\}$ instead of the whole $\RR^{N_x}$, and cannot be directly solved by SOSTOOLS.
In this work, we derive a new SOSP formulation tailored to \Cref{eq:synthesis}.
The core idea behind SOSP for enforcing a condition is to construct a \textit{refute set} which \revcolor{violates the condition and} show that the set is empty.
The emptiness can be shown using \textit{Positivstellensatz} \cite{parrilo2003sosp}, which uses \textit{ring-theoretic cone} \cite{bochnak2013real} as explained below.


\begin{definition}[Ring-theoretic cone]
    Denote $\RR[x_1,\dots,x_n]$ a set of polynomials with $[x_1,\dots,x_n]$ as variables.
    For a set $S\defeq\{\gamma_1,\dots,\gamma_s\}\subseteq\RR[x_1,\dots,x_n]$, the associated ring-theoretic cone is given as:
    $\Gamma = \{p_0+p_1\gamma_1+\dots+p_s\gamma_s+\dots+p_{12\dots s}\gamma_1\cdots\gamma_s\}$
    where $p_0,p_1,\dots,p_{12\dots s}\in SOS$.
\end{definition}


\begin{theorem}[Positivstellensatz]\label{thm:pos}
    Let $(\gamma_j)_{j=1,\dots,s}$, $(\psi_k)_{k=1,\dots,t}$, $(\zeta_l)_{l=1,\dots,r}$ be finite families of polynomials in $\RR[x_1,\dots,x_n]$.
    Let $\Gamma$ be the ring-theoretic cone generated by $(\gamma_j)_{j=1,\dots,s}$, $\Psi$ the multiplicative monoid \cite{parrilo2003sosp} generated by $(\psi_k)_{k=1,\dots,t}$, and $Z$ the ideal \cite{parrilo2003sosp} generated by $(\zeta_l)_{l=1,\dots,r}$.
    Then, the following properties are equivalent:
    \begin{itemize}
        \item The following set is empty
        \begin{equation}\label{eq:pos_refute}
            \left\{
                x\in\RR^n \left|
                \begin{aligned}
                    ~&\gamma_j(x) \geq 0, \forall j=1,\dots,s \\
                    ~&\psi_k(x) \neq 0, \forall k=1,\dots,t \\
                    ~&\zeta_l(x) = 0, \forall l=1,\dots,r
                \end{aligned}\right.
            \right\}
        \end{equation}
        \item There exists $\gamma\in\Gamma$, $\psi\in\Psi$, $\zeta\in Z$ such that
        \begin{equation}\label{eq:pos_feas}
            \gamma + \psi^2 + \zeta = 0 \mathrm{~where~} \psi = 1 \mathrm{~if~} t=0.
        \end{equation}
    \end{itemize}
\end{theorem}

The proof of \Cref{thm:pos} can be found in \cite{parrilo2003sosp}.
Intuitively, Positivstellensatz allows us to transform the problem of showing that a refute set is empty into a feasibility problem.
We will construct a refute set for \eqref{eq:synthesis} in the form of \eqref{eq:pos_refute}, and then prove it is empty by solving a feasibility problem \eqref{eq:pos_feas}.

\subsection{The Local Manifold Positiveness Problem}

We first transform \Cref{problem:synthesis} into a local manifold positiveness problem.
Following \eqref{eq:synthesis}, we have
\begin{align}
    &\minimizewrt{u\in\cU(x)} \dot{\phi}_\theta(x,u) = \minimizewrt{u\in\cU(x)} \underbrace{\frac{\partial\phi_\theta}{\partial x}f(x)}_{L_f\phi_\theta(x)} + \underbrace{\frac{\partial\phi_\theta}{\partial x}g(x)}_{L_g\phi_\theta(x)} u \\
    = &\minimizewrt{u\in\cU(x)} L_f\phi_\theta + \textstyle\sum_{i=1}^{N_u} L_g\phi_\theta[i]u[i] \\
    = &L_f\phi_\theta + \textstyle\sum_{i=1}^{N_u} L_g\phi_\theta[i] \delta_{L_g\phi_\theta[i]}\left(\underline{u}(x)[i], \bar{u}(x)[i]\right)
\end{align}
where $\delta_{c}(a,b)=a$ if $c\geq 0$ otherwise $b$.
Hence, we acquire:
\begin{problem}[Local Manifold Positiveness]\label{problem:localpos}
    Find $\theta\in\Theta$ such that $\forall x\in\cX ~ \st ~ \phi_\theta(x)\geq 0$,
    \begin{equation}
    L_f\phi_\theta + \textstyle\sum_{i=1}^{N_u} L_g\phi_\theta[i] \delta_{L_g\phi_\theta[i]}\left(\underline{u}(x)[i], \bar{u}(x)[i]\right) < -\eta
    \end{equation}
    where we focus on the \textbf{local manifold} $\cX\cap\{x\mid\phi_\theta(x)\geq 0\}$.
\end{problem}

\begin{remark}
    It is sufficient to consider only the local manifold to achieve the safe control objectives in \Cref{sec:formulation_safe_control}.
    That also yields a less conservative (but still feasible) safety index than a global positiveness problem would give \cite{zhao2023sos}.
\end{remark}

\subsection{The Refute Problem}
We now constructing a refute set for \Cref{problem:localpos} as:
\begin{equation}\label{eq:refute_set}
    \begin{cases}
        L_f\phi_\theta + \sum_{i=1}^{N_u} L_g\phi_\theta[i] \delta_{L_g\phi_\theta[i]}\left(\underline{u}(x)[i], \bar{u}(x)[i]\right) \geq -\eta \\
        x \in \cX, ~ \phi_\theta(x) \geq 0
    \end{cases}
\end{equation}

It can be shown that finding $\theta$ to make \eqref{eq:refute_set} empty solves \Cref{problem:localpos}.
Let $I^+,I^-\subseteq[N_u]$ denote the set of indices of non-negative and negative elements in $L_g\phi_\theta$ (hence $I^+\cap I^-=[N_u]$), we can re-write \eqref{eq:refute_set} as:
\begin{equation}\label{eq:refute_set_I}
    \begin{cases}
        \gamma_1 \defeq L_f\phi_\theta + \sum_{i\in I^+} L_g\phi_\theta[i] \underline{u}(x)[i]\\
        ~~~~~~~~~~~~~~+ \sum_{i\in I^-} L_g\phi_\theta[i] \bar{u}(x)[i] \geq -\eta \\
        \gamma_{1+i} \defeq h_i(x) \geq 0,~\forall i=1,\dots,N_h \\
        \gamma_{1+N_h+i} \defeq L_g\phi_\theta(x)[i] \geq 0,~\forall i\in I^+ \\
        \gamma_{1+N_h+|I^+|+i} \defeq -L_g\phi_\theta(x)[i] \geq 0,~\forall i\in I^- \\
        \gamma_{2+N_h+N_u} \defeq \phi_\theta(x) \geq 0 \\
        \zeta_i(x) = 0,~\forall i=1,\dots,N_\zeta
    \end{cases}
\end{equation}
where $\{h_i\}_i$ define the state space $\cX$ (see \Cref{sec:dynamic_system}) and $\{\zeta_i\}_i$ are auxiliary functions to be constructed on demand.
\revcolor{
Note that $\gamma_i$ and $\zeta_i$ are functions of $x$.
}
For example, if the state $x$ contains elements $x[i]=\sin\alpha$ and $x[j]=\cos\alpha$ for some angle $\alpha$, one can employ $\zeta_l(x) = x[i]^2+x[j]^2-1$.
We need to consider all possible values of $L_g\phi_\theta$ and construct the set \eqref{eq:refute_set_I} for each of $2^{N_u}$ possible assignments of $I^+$ and $I^-$.
The refute problem is then formalized as \Cref{problem:refute}.

\begin{problem}[Refute]\label{problem:refute}
    Find $\theta\in\Theta$ such that $\forall I^+\subseteq [N_u],I^-=[N_u]\setminus I^+$, $\nexists x\in\RR^{N_x}$ such that \eqref{eq:refute_set_I} holds.
\end{problem}

\begin{remark}
    \Cref{problem:refute} shows that at most $2^{N_u}$ refute sets are empty.
    In practice, $L_g\phi_\theta(x)$ is limited by the bounded state space $\cX$, which renders certain values of $(I^+,I^-)$ impossible and reduces the refute sets to prove empty.
\end{remark}

\subsection{The Nonlinear Programming Problem}

Since all conditions in \eqref{eq:refute_set_I} correspond to the positiveness conditions in \eqref{eq:pos_refute}, \Cref{thm:pos} allows us to solve \Cref{problem:refute} by reposing it as a feasibility problem.
Let $N\defeq 2+N_h+N_u$, we find $\theta\in\Theta$, $p'_i\in\RR[x]$, and $p_i\in SOS,~\forall i$ such that,
\begin{equation}\label{eq:nonlinear_pos}
    \begin{aligned}
        &\gamma = p_0 + p_1\gamma_1 + p_2\gamma_2 + \dots + p_N\gamma_N + \\
        &~~~~~~~~p_{12}\gamma_{1}\gamma_{2} + \dots + p_{12\dots N}\gamma_{1}\dots\gamma_{N} \\
        &\psi = 1,~\zeta = \textstyle\sum_{i=1}^{N_\zeta}p'_i\zeta_i,~\gamma + \psi^2 + \zeta = 0
    \end{aligned}
\end{equation}
for each possible $(I^+,I^-)$.
To efficiently solve the problem, we set $p_i,p'_i\in\RR, p_i\geq 0$ for $i>0$, and rewrite \eqref{eq:nonlinear_pos} as
\begin{equation}\label{eq:nonlinear_sos}
    \begin{aligned}
        &\exists p_i, p'_i \in \RR, p_i \geq 0 ,~ \forall i>0 ~ \st \\
        &p_0 = -1 - \textstyle\sum_{i=1}^{N_\zeta}p'_i\zeta_i - p_1\gamma_1 - p_2\gamma_2 - \dots - p_N\gamma_N \\
        &~~~~~~~~ -p_{12}\gamma_{1}\gamma_{2} - \dots - p_{12\dots N}\gamma_{1}\dots\gamma_{N} \in SOS
    \end{aligned}
\end{equation}
\revcolor{
Note that $p_0$ now becomes a polynomial of $x$ where the coefficients contain decision variables $\theta$, $p_i$, and $p_i'$.
}
Our goal is to find real-valued coefficients such that $p_0$ is a sum of squares polynomial of $x$.
One possible way is to decompose $p_0$ in quadratic form with a positive semi-definite coefficient matrix.
Namely, assume $p_0$ has degree $2d$ and let $\bx\defeq\left[1,x[1],\dots,x[N_x],x[1]x[2],\dots,x[N_x]^d\right]^\top$, $\bp\defeq [p'_1,\dots,p'_{N_\zeta}, p_1,p_2,\dots,p_{012\dots N}]^\top$, the decomposition is given by $p_0 = \bx^\top Q(\theta,\bp)\bx$.
$Q(\theta,\bp)$ is a symmetric matrix with $Q[i,i]$ to be the coefficient of $\bx[i]^2$ in $p_0$ and $Q[i,j]$ ($i\neq j$) to be half of the cofficient of $\bx[i]\bx[j]$ in $p_0$.
It can be shown that if $Q(\theta,\bp)\succeq 0$, $p_0$ is an SOS polynomial and \eqref{eq:nonlinear_sos} holds.
Let $Q_i(\theta,\bp_i)$ denote the matrix from the condition \eqref{eq:nonlinear_sos} generated by the $i^\mathrm{th}$ assignment of $(I^+, I^-)$.
\Cref{problem:refute} is equivalent to the following problem:
\begin{problem}[Nonlinear Programming]\label{problem:nonlinear}
    Find $\theta\in\Theta$ and $\bp_1,\bp_2,\dots,\bp_{2^{N_u}}$ where $\forall i,~ \bp_i[j] \in \RR$ for $j > 0$ and $\bp_i[j] \geq 0$ for $j > N_\zeta$, such that $\forall i,~Q_i(\theta, \bp_i) \succeq 0$.
\end{problem}
\revcolor{Due to the product of $p_i$ and $\gamma_i$ (which contains $\theta$), constraints $Q_i(\theta, \bp_i) \succeq 0$ in are non-convex.}
\revcolor{
We will use general solvers such as} \verb|fmincon| \revcolor{and check the PSD constraints via eigenvalue decomposition}.
\revcolor{
We leave speciality solutions leveraging the structure of $Q_i$ as future work.
}

\begin{remark}
\revcolor{
    One limitation of ours is the scalability.
    In \Cref{problem:nonlinear}, there are two nested exponential complexities with respect to $N_u$ and $N_h$, mainly due to the poor scalability of Positivstellensatz.
    One can alleviate the complexity by computing eigenvalues of $2^{N_u}$ $Q_i$'s in parallel.
    In \Cref{sec:exp}, we will show that \Cref{problem:nonlinear} can be solved by current solvers.
}
\end{remark}

\section{Theoretical Results}\label{sec:theory}

In this section, we formally prove that solving \Cref{problem:nonlinear} yields a feasible safety index described in \Cref{sec:formulation_safe_control}.
We first prove that solutions to \Cref{problem:synthesis} \revcolor{guarantee both} forward invariance within a safe set $\cX_\mathrm{safe}$ and finite-time convergence to that set, then prove that the solution to \Cref{problem:nonlinear} also solves \Cref{problem:synthesis}.
In literature, \cite{liu2014control} provides a preliminary proof of FI given a feasible safety index (e.g., a solution to \Cref{problem:synthesis}).
The proof however leverages an unproved key assumption\footnote{\cite{liu2014control} also proves FI by induction, but based on an unjustified assumption: if a state $x$ is not reachable under \eqref{eq:safe_control_law}, we have $\phi_n(x) > 0$. That assumption depends on the characterization of the reachable set which is absent.}.
In this work, we prove FI without that assumption.
In addition, the SIS approach by \cite{zhao2023sos} only considers a special case of \Cref{problem:synthesis}, lacking the FTC.
To proceed, we first define the $m^\mathrm{th}~(m\in[0,n])$ principal set as the intersection of zero sub-level sets of the $m+1$ highest order safety indices defined by $\phi_n$.

\begin{definition}[Principal Set $\cX_m(\phi_n)$]
    Given a safety index $\phi_n = \prod_{i=1}^{n}(1+a_i s)\phi_0$ ($n\geq 0$), the $m^\mathrm{th}~(m\in[0,n])$ principal set $\cX_m(\phi_n) \defeq \{x \mid \phi_j(x) \leq 0,~\forall j=n-m,\dots,n\}$ where $\phi_j = \prod_{i=1}^{j}(1+a_i s)\phi_0$.
\end{definition}

We define the safe set $\cX_\mathrm{safe}\equiv \cX(\phi_n)$ within which the system will be forward invariant for an $n^\mathrm{th}$ order $\phi_n$.

\begin{definition}[Safe Set $\cX(\phi_n)$]\label{def:safe_set}
    Given a safety index $\phi_n$ ($n\geq 0$), we define the safe set $\cX(\phi_n)$ as the $n^\mathrm{th}$ principal set, i.e., $\cX(\phi_n)\defeq \cX_n(\phi_n)$.
    It follows that $\cX(\phi_0) \equiv \cX_S$.
\end{definition}


We now prove that each of the principal sets is invariant before presenting the core theorems of this paper.

\begin{lemma}[Forward Invariance of Principal Sets]\label{lma:FIPS}
    If a feasible safe control law $c_{\phi_n}$ as in \Cref{def:feasible_control_law} is applied, then the principal set $\cX_m(\phi_n)$ is invariant for $m\in[0,n]$.
    Namely, $\forall x(0) \in \cX_m(\phi_n) \rightarrow x(t) \in \cX_m(\phi_n)$ for $t>0$.
\end{lemma}
\begin{proof}
    We write $\phi(t)\equiv\phi(x(t))$ for clarity.
    For any $x(0) \in \cX_m(\phi_n)$, we have $\phi_j(0) \leq 0,~\forall j\in[n-m,n]$.
    We will show $\phi_j(t) \leq 0$ for $t>0,~j\in[n,n-m]$ by induction.
    
    Starting with $j=n$, note that $\phi_0$ is $n+1^\mathrm{th}$ order differentiable, making $\phi_n$ differentiable and also continuous.
    Since $\phi_n$ solves \Cref{problem:synthesis}, the safe control law \eqref{eq:safe_control_law} is always feasible, e.g., $\dot{\phi}_n(\tau) < 0$ when $\phi_n(\tau) \geq 0$ for $\tau\geq 0$.
    To prove by contradiction, assume $\phi_n(0) \leq 0$, but $\phi_n(t) > 0$.
    Find $[t', t]$ as the maximal connected set in $\phi_n^{-1}([0,\infty))\cap [0, t]$.
    Then, $\phi_n(t') = 0$ and for $\tau\in[t', t]$, $\phi_n(\tau) \geq 0 \rightarrow \dot{\phi}_n(\tau) < 0$.
    Since $\phi_n$ is differentiable, $\phi_n(t) = \phi_n(t') + \int_{t'}^t \dot{\phi}_n(\tau) d\tau < \phi_n(t') = 0$ which contradicts with the assumed $\phi_n(t) > 0$.
    Hence, $\phi_n(t) \leq 0$.
    
    For $n \geq j > n-m$, we have $\phi_{j} = \phi_{j-1} + a_j\dot{\phi}_{j-1}$.
    $\forall t\geq 0$, assume $\phi_{j}(t) \leq 0$,
    we have $\phi_j(\tau) \leq 0$ for $\tau\in[0, t]$.
    To prove by contradiction, assume $\phi_{j-1}(t) > 0$.
    Find $[t', t]$ as the maximal connected set in $\phi_{j-1}^{-1}([0, \infty]) \cap [0, t]$.
    Since $\phi_{j-1}(0) \leq 0$, we have $\phi_{j-1}(t')=0$ and $\phi_{j-1}(\tau) \geq 0$ for $\tau\in[t', t]$.
    Since $a_j > 0$, so $\dot{\phi}_{j-1}(\tau) \leq 0$ in $[t', t]$.
    Then, $\phi_{j-1}(t) = \phi_{j-1}(t') + \int_{t'}^{t} \dot{\phi}_{j-1}(\tau) d\tau \leq \phi_{j-1}(t') = 0$
    This contradicts with the assumed $\phi_{j-1}(t) > 0$.
    Hence $\forall t>0,~\phi_j(t) \leq 0 \rightarrow \phi_{j-1}(t) \leq 0$ which completes the induction.
    Hence $x(t) \in \cX_m(\phi_n)$ and $\cX_m(\phi_n)$ is invariant.
\end{proof}

\begin{theorem}[Forward Invariance within the Safe Set]\label{thm:FISSA}
    If a feasible safe control law $c_{\phi_n}$ in \Cref{def:feasible_control_law} is applied, then the safe set $\cX(\phi_n)$ in \Cref{def:safe_set} is invariant and also a subset of $\cX_S$.
\end{theorem}

\begin{proof}
    Since the safe set $\cX(\phi_n)$ is the $n^\mathrm{th}$ principal set, i.e., $\cX(\phi_n)\defeq \cX_n(\phi_n)$, $\cX(\phi_n)$ is invariant by \Cref{lma:FIPS}
    By definition, $\forall x\in\cX(\phi_n), \phi_0(x) \leq 0 \rightarrow x\in \cX_S$, hence $\cX(\phi_n)$ is a subset of $\cX_S$.
\end{proof}

\begin{remark}
    Both SSA (this work) and HOCBF \cite{xiao2019control} prove forward invariance, but based on different control constraints.
    SSA enforces a constant bound on $\dot{\phi}_n$, while that bound in HOCBF depends on the barrier function value.
\end{remark}

\begin{theorem}[Finite-time Convergence to the Safe Set]\label{thm:FCSSA}
    If a feasible safe control law $c_{\phi_n}$ (\Cref{def:feasible_control_law}) is applied, the system will land in $\cX(\phi_n)$ in finite time and stay within $\cX(\phi_n)$, i.e., $\exists t\in[0, \infty) \rightarrow x(\tau)\in\cX(\phi_n),~\forall \tau\geq t$.
\end{theorem}
\begin{proof}
    We consider whether the system starts in $\cX(\phi_n)$ or not.
    If $x(0)\in\cX(\phi_n)$, we have $x(\tau)\in\cX_S$ for $\tau\geq 0$ since $\cX(\phi_n)\subseteq\cX_S$ is invariant by \Cref{thm:FISSA}.
    If $x(0)\notin\cX(\phi_n)$, we will first prove that $\exists t_m \in [0,\infty)\rightarrow x(t_m)\in\cX_m(\phi_n)$ for $m=0,\dots,n$ by induction.
    
    Starting with $m=0$ case.
    If $\phi_n(0) \leq 0$, we take $t_0=0$ and $x(t_0)=x(0)\in\cX_0(\phi_n)$.
    If $\phi_n(0) > 0$, since $c_{\phi_n}$ is feasible, $\dot{\phi}_n(t) < 0$ always holds whenever $\phi_n(t) \geq 0$.
    Since $\phi_n$ strictly decreases whenever it is positive and $\phi_n(0)<\infty$, $\exists t_0 \in [0,\infty) \rightarrow \phi_n(t_0) \leq 0 \rightarrow x(t_0)\in\cX_0(\phi_n)$.

    For $0\leq m < n$, assume $\exists t_m \rightarrow x(t_m) \in \cX_m(\phi_n)$.
    By \Cref{lma:FIPS}, $\cX_m(\phi_n)$ is invariant.
    Hence $\forall t\in[t_m,\infty)$, $\phi_j(t) \leq 0$ for $j=n-m,\dots,n$.
    If $\phi_{n-m-1}(t_m)\leq 0$, then take $t_{m+1} = t_m$ and we have $x(t_{m+1}) = x(t_m)\in \cX_{m+1}(\phi_n)$.
    If $\phi_{n-m-1}(t_m) > 0$, recall $\forall t\geq t_m$, $\phi_{n-m}(t) = \phi_{n-m-1}(t) + a_{n-m} \dot{\phi}_{n-m-1}(t) \leq 0$ and $a_{n-m} > 0$.
    Then, $\forall t\geq t_m$, whenever $\phi_{n-m-1}(t)> 0$, we have $\dot{\phi}_{n-m-1}(t) \leq -\frac{1}{a_{n-m}}\phi_{n-m-1}(t) < 0$ making $\phi_{n-m-1}(t)$ strictly decrease when it is positive.
    Hence, $\exists t_{m+1} > t_m \rightarrow \phi_{n-m-1}(t_{m+1}) \leq 0 \rightarrow x(t_{m+1})\in\cX_{m+1}(\phi_n)$ which completes the induction.
    It follows that $\exists t_n \in [0,\infty) \rightarrow x(t_n) \in \cX_n(\phi_n) \equiv \cX(\phi_n)$.
    Finally, $x(\tau)\in\cX(\phi_n)$ for $\tau\geq t_n$ since $\cX(\phi_n)$ is invariant by \Cref{thm:FISSA}.
\end{proof}

\begin{remark}
    We proved that by applying constraints on only $\phi_n$, all lower order safety indices $\phi_j$ ($j\in[0,n]$) will reduce to non-positive values in finite-time and stay that way.
    HOCBF (no FTC) \cite{xiao2019control} requires the state to start in the safe set, while our approach allows unsafe initial states.
    \revcolor{
    Besides, at the safety boundary $\phi=0$, our approach enforces a strict decrease of $\phi$, while HOCBF only requires $\dot{\phi} \leq 0$, allowing the system to stay on the boundary.
    Hence, our approach leads to a more restrictive safety behavior than HOCBF.
    }
\end{remark}

We then present the main theory of this paper as follows:
\begin{theorem}[Feasibility of Nonlinear Programming]\label{thm:feas_nonlinear}
    Suppose we are given an $n+1^\mathrm{th}$ order differentiable function $\phi_0$.
    Let $\theta^*=[k^*_1,\dots,k^*_n],\{\bp^*_i\}_{i=1,\dots,2^{N_u}}$ be the solution to \Cref{problem:nonlinear}, then the feasible control law $c_{\phi_{\theta^*}}$ in \Cref{def:feasible_control_law} achieves both forward invariance within the safe set $\cX(\phi_n)$ and finite-time convergence to the safe set $\cX(\phi_n)$.
\end{theorem}

\begin{proof}
    Since $\theta^*=[k^*_1,\dots,k^*_n],\{\bp^*_i\}_{i=1,\dots,2^{N_u}}$ solves \Cref{problem:nonlinear}, then $\forall (I^+,I^-)_i,Q_i(\theta^*,\bp^*_i)\succeq 0$.
    With EVD, we have
    \begin{align}
        p_{0, i} &= \bx^\top Q_i(\theta^*,\bp^*_i)\bx = \bx^\top V_i \Sigma_i V_i^\top \bx \\
        &= \left(\Sigma_i^{1/2}V_i^\top \bx\right)^\top\Sigma_i^{1/2}V_i^\top \bx \in SOS.
    \end{align}
    Hence, \eqref{eq:nonlinear_sos} holds, and equivalently the feasibility condition \eqref{eq:nonlinear_pos} holds.
    Using \Cref{thm:pos}, we know the refute set \eqref{eq:refute_set_I} is empty, and equivalently the set \eqref{eq:refute_set} is empty.
    Then, for all $x$ such that $x\in\cX$ and $\phi_{\theta^*}(x)\geq 0$, we have
    \begin{align}
        &\minimizewrt{u\in\cU(x)} \dot{\phi}_{\theta^*}(x,u)=L_f\phi_{\theta^*} + \\
        &~~~~~~\textstyle\sum_{i=1}^{N_u} L_g\phi_{\theta^*}[i] \delta_{L_g\phi_{\theta^*}[i]}\left(\underline{u}(x)[i], \bar{u}(x)[i]\right) < -\eta
    \end{align}
    which means $\phi_{\theta^*}$ solves \Cref{problem:synthesis} and the control law $c_{\phi_{\theta^*}}$ is feasible as in \Cref{def:feasible_control_law}.
    By \Cref{thm:FISSA} and \Cref{thm:FCSSA}, $c_{\phi_{\theta^*}}$ achieves both forward invariance within the safe set $\cX(\phi_n)$ and finite-time convergence to the safe set $\cX(\phi_n)$.
\end{proof}

\section{Numerical Study}\label{sec:exp}

We apply the proposed SIS approach to a second-order unicycle model with state-dependent control limits.
Applications to other systems such as integrators are similar.

\subsection{Unicycle with State-dependent Control Limits}

Consider a unicycle (see \Cref{fig:agent}) with distance $d$ to a static obstacle\footnote{If there are multiple obstacles, we consider the closest one at each time step. We leave the joint consideration of multiple obstacles as future work.}, longitudinal velocity $v$, relative heading $\alpha$ (with respect to the obstacle), and azimuth angle $\beta$.
\begin{wrapfigure}{r}{0.24\linewidth}
    \vspace{-12pt}
    \centering
    \includegraphics[width=0.9\linewidth]{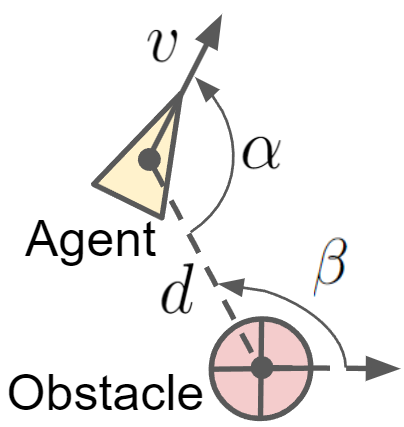}
    \caption{Unicycle.}
    \vspace{-12pt}
    \label{fig:agent}
\end{wrapfigure}
The control includes the acceleration $a\equiv \dot{v} \in[(v_\mathrm{min}-v)/dt, (v_\mathrm{max}-v)/dt]$ and relative heading velocity\footnote{The change of $\alpha$ depends on all of $v$, $d$, and the global heading velocity of the agent. We formulate this way to simplify the derivation of SIS.} $w\equiv \dot{\alpha} \in [w_\mathrm{min}, w_\mathrm{max}]$.
$dt$ is the time interval.
With the state-dependent limits on $a$, the velocity is effectively bounded within $[v_\mathrm{min}, v_\mathrm{max}]$.
With state $x\defeq [d,v,\alpha,\beta]^\top$ and control $u\defeq [a,w]^\top$, the overall dynamics is given by
\begin{equation}\label{eq:example_unicycle}
    \dot{x} = 
    \begin{bmatrix}
        \dot{d} \\ \dot{v} \\ \dot{\alpha} \\ \dot{\beta}
    \end{bmatrix} = \begin{bmatrix}
        -v\cos\alpha \\ 0 \\ 0 \\ -vd^{-1}\sin\alpha
    \end{bmatrix} + \begin{bmatrix}
        0 & 0 \\ 1 & 0 \\ 0 & 1 \\ 0 & 0
    \end{bmatrix}
    \begin{bmatrix}
        a \\ w
    \end{bmatrix}
\end{equation}
With $\phi_0 = d_\mathrm{min}-d$, the goal is to synthesis a safety index $\phi_\theta = \phi_0 + k \dot{\phi}_0$ such that the control law \eqref{eq:safe_control_law} always keeps the agent at least $d_\mathrm{min}$ away from the obstacle.

\subsection{SIS Problem}

By \Cref{thm:feas_nonlinear}, the goal above can be achieved by solving \Cref{problem:nonlinear} which is transformed from \Cref{problem:synthesis}.
Hence, we start with the feasibility condition \eqref{eq:synthesis} for the system \eqref{eq:example_unicycle}:
\begin{equation}\label{eq:example_feas}
\begin{aligned}
    \forall d,\alpha,v\in[v_\mathrm{min}, v_\mathrm{max}], d_\mathrm{min}-d+kv\cos\alpha \geq 0, \\
    \minimizewrt{a,w} v\cos\alpha + ka\cos\alpha - kvw\sin\alpha < -\eta.
\end{aligned}
\end{equation}
With $x\defeq \sin\alpha$, $y\defeq \cos\alpha$, and $z\defeq v$, we can construct the refute set for \eqref{eq:example_feas} as:
\begin{equation}\label{eq:example_refute}
    \begin{cases}
        \gamma_1\defeq yz-k\tilde{w}xz+\frac{1}{dt}k\tilde{v}y-\frac{1}{dt}kyz+\eta \geq 0 \\
        \gamma_2\defeq \II_{xz}xz \geq 0, \gamma_3\defeq \II_{y}y \geq 0, \gamma_4\defeq \II_{x}x \geq 0 \\
        \gamma_5\defeq \II_{x}\II_{y}xy \geq 0 \\
        \gamma_6\defeq -z^2+(v_\mathrm{min}+v_\mathrm{max})z-v_\mathrm{min}v_\mathrm{max} \geq 0 \\
        \gamma_7\defeq kyz+d_\mathrm{min} \geq 0 \\
        \xi \defeq x^2 + y^2 - 1 = 0 \\
    \end{cases}
    \vspace{-1pt}
\end{equation}
for $\II_{xz}=\pm 1$, $\II_{x}=\pm 1$, $\II_{y}=\pm 1$.
$\tilde{w}=w_\mathrm{max}$ if $\II_{xz}=1$ and $\tilde{w}=w_\mathrm{min}$ if $\II_{xz}=-1$.
$\tilde{v}=v_\mathrm{min}$ if $\II_{y}=1$ and $\tilde{v}=v_\mathrm{max}$ if $\II_{y}=-1$.
Hence, we have eight refute sets \eqref{eq:example_refute} to prove empty based on the sign value of $\II_{*}$.
Following \eqref{eq:nonlinear_sos}, for the $i^\mathrm{th}$ assignment ($i\in[8]$) of $(\II_{xz}, \II_{x}, \II_{y})$, we can construct
\begin{equation}\label{eq:nonlinear_sos_simp}
\begin{aligned}
    p_{i,0} = -1-p'_i\xi-p_{i,1}\gamma_{i,1}-\cdots-p_{i,7}\gamma_{i,7}
\end{aligned}
\end{equation}
and decompose as $p_{i,0}= \bx^\top Q_i(\theta,\bp_i)\bx$ where $\bx\defeq[1, x, y, z]^\top$, $\theta\defeq[k]$, and $\bp_i\defeq [p'_i,p_{i,1},\dots,p_{i,7}]$.
By \Cref{thm:feas_nonlinear}, we just need to find $k\geq 0$, $p'_i\in\RR$, $p_{i,1},\dots,p_{i,7}\geq 0$ such that $Q_i(\theta,\bp_i)\succeq 0$ for $i\in[8]$ which can be encoded as a nonlinear programming problem.
\revcolor{
    \eqref{eq:nonlinear_sos_simp} has only the ``first-order'' terms of $p$ which is a valid simplification of \eqref{eq:nonlinear_sos} when $\cX$ is bounded \cite{putinar1993positive}.
    With unbounded $\cX$, \eqref{eq:nonlinear_sos_simp} is generally sufficient but not necessary for \eqref{eq:nonlinear_sos} and saves computation empirically.
    If \eqref{eq:nonlinear_sos_simp} fails, one needs to try the full form \eqref{eq:nonlinear_sos}.
}

\begin{table}[t]
    \centering
    \caption{SIS and Safe Control Results}
    \begin{tabular}{cccc}
    \toprule
         $k$ & Time (\revcolor{s}) & Safe Set ($\%$) & Violations \\
    \midrule
         \revcolor{$0.0100 \pm 3.96e-06$} & \revcolor{24.955 $\pm$ 6.295} & 100.0 $\pm$ 0.0 & 0.0 $\pm$ 0.0 \\
    \bottomrule
    \end{tabular}
    \label{tab:example_res}
    \vspace{-20pt}
\end{table}

\vspace{-15pt}
\subsection{Results}

We set $v_\mathrm{min}=w_\mathrm{min}=-1$, $v_\mathrm{max}=w_\mathrm{max}=1$, $\eta=0.001$ and design the agent to avoid the radius of $d_\mathrm{min}=1$ within the origin.
We implement \Cref{problem:nonlinear} in MATLAB using \verb|fmincon| and apply the PSD condition by enforcing non-negative eigenvalues of $Q_i$ with tolerance $1e-6$.
\revcolor{The eigenvalues are computed in parallel for each $Q_i$ using} \verb|parfor|.
For each synthesized $\phi_n$, we apply the control law \eqref{eq:safe_control_law} to a simulated navigation task with random initial state and goal locations and set $dt=0.01$ s.
We employ a proportional controller to generate the nominal control $u^r$ and find a safe control $u$ with minimal deviation (i.e., $\cJ(u) = \|u-u^r\|$).
Considering numerical errors, we randomize the variable initialization and repeat for $10$ times.

We report in \Cref{tab:example_res} the mean and standard deviation of the solution $k$, SIS solving time, whether the agent lands in the safe set (see \Cref{def:safe_set}) before the task ends, and the number of safety violations ($\phi_0>0$) since entering the safe set per task trial.
It is evident that our SIS approach can robustly generate safe control laws that satisfy finite-time convergence ($100\%$ landing in the safe set) and forward invariance ($0$ safety violations) with respect to the safe set.

\section{Conclusion}\label{sec:conclusion}

This paper introduced a safety index synthesis (SIS) approach to finding feasible safe control laws.
We leveraged Positivstellensatz to transform the SIS problem with an infinite number of constraints into nonlinear programming.
The approach was backed with a formal proof of forward invariance within the safe regions, aided by a new characterization of the invariant safe set.
We also proved the finite-time convergence guarantee which was missing in previous work.
The proposed method was applied to a simulated unicycle and achieved robust computation with zero safety violations.

Regarding extensions, it is reasonable to consider nonlinear terms in the general form of the safety index (see \Cref{sec:formulation_safe_control}) while preserving FI and FTC.
One can also extend to consider piece-wise differentiable or even non-differentiable safety indices $\phi_0$ which can emerge in cases such as avoiding multiple obstacles.

\addtolength{\textheight}{-12cm}   








\bibliographystyle{IEEEtran}
\bibliography{ref}

\end{document}

%% file: math_macros.tex
\newcommand{\bp}{\bm{p}}

\newcommand{\bx}{\bm{x}}

\newcommand{\cJ}{\mathcal{J}}

\newcommand{\cU}{\mathcal{U}}

\newcommand{\cX}{\mathcal{X}}

\newcommand{\II}{\mathbb{I}}

\newcommand{\RR}{\mathbb{R}}

\newcommand{\minimize}{\mathop{\mathbf{min}}}

\newcommand{\minimizewrt}[1]{\mathop{\underset{#1}{\minimize}}}

\newcommand{\defeq}{\vcentcolon=}

\newcommand{\st}{\mathop{\mathbf{s.t.}}}

\newtheorem{problem}{Problem}
\newtheorem{theorem}{Theorem}

\newtheorem{lemma}[theorem]{Lemma}
\newtheorem*{remark}{Remark}
\newtheorem{definition}{Definition}

\colorlet{revision_color}{black}
\newcommand{\revcolor}[1]{\textcolor{revision_color}{#1}}